\documentclass[11pt]{article}
\usepackage{color}
\usepackage[colorlinks=true,citecolor=blue]{hyperref}
\usepackage{cite}
\usepackage{amssymb,amsmath,amsfonts,amsthm,amsthm}
\usepackage{mathtools}
\usepackage[english]{babel}
\usepackage[utf8]{inputenc}
\usepackage{fullpage}
\usepackage{graphicx}
\usepackage{times,fourier}

\makeatletter
\renewcommand\section{\@startsection {section}{1}{\z@}%
	{-2ex \@plus -1ex \@minus -.2ex}%
	{1ex \@plus.1ex}%
	{\normalfont\bf\sffamily}}
\renewcommand\subsection{\@startsection{subsection}{2}{\z@}%
	{-1.75ex\@plus -0.4ex \@minus -.2ex}%
	{0.6ex \@plus .1ex}%
	{\normalfont\small\bf\sffamily}}
\renewcommand\subsubsection{\@startsection{subsubsection}{3}{\z@}%
	{-0.6ex\@plus -0.2ex \@minus -.2ex}%
	{0.4ex \@plus .1ex}%
	{\normalfont\normalsize\it}}
\renewcommand\paragraph{\@startsection{paragraph}{4}{\z@}%
	{0.2ex \@plus0.2ex \@minus0.1ex}{-0.5em}%
	{\normalfont\normalsize\bfseries}}
\def\ps@headings{%
	\let\@oddfoot\@empty
	\let\@evenfoot\@empty
	\def\@evenhead{\small\sffamily\thepage\hfil\slshape\leftmark}%
	\def\@oddhead{\small\sffamily{\slshape\rightmark}\hfil\thepage}%
	\let\@mkboth\markboth
	\def\chaptermark##1{\markboth{{\ifnum \c@secnumdepth >\m@ne
				\if@mainmatter \@chapapp\ \thechapter. \ \fi \fi ##1}}{}}%
	\def\sectionmark##1{\markright {{\ifnum \c@secnumdepth >\z@
				\thesection. \ \fi ##1}}}}
\def\fbf#1{\setbox0=\hbox{$#1$}\kern-0.10\wd0
	\lower0.02em\copy0\kern-\wd0 \lower0.02em\hbox{\kern+0.04em\copy0}\kern-\wd0
	\raise0.00em\copy0\kern-\wd0 \raise0.00em\hbox{\kern-0.04em\box0}}
\makeatother

\advance\textwidth -4mm
\advance\hoffset 2mm
\advance\textheight 0mm
\advance\voffset 2mm
\advance\parskip 4pt
\numberwithin{equation}{section}
1

\newtheorem{theorem}{Theorem}[section]

\newtheorem{corollary}[theorem]{Corollary}
\newtheorem{lemma}[theorem]{Lemma}
\newtheorem{remark}[theorem]{Remark}
\newtheorem{proposition}[theorem]{Proposition}

\def\maketitle{\par\noindent{\LARGE\bf\sffamily\thetitle}\\[1.4ex]
	{\large\theauthor}\\[0.6ex]
	\textit{\thetextinfo}\\[0.2ex]
	{\small\today}\par\vglue1.4\bigskipamount}
\def\title#1{\def\thetitle{#1}}
\def\author#1{\def\theauthor{#1}}
\def\textinfo#1{\def\thetextinfo{#1}}

\def\be{\begin{equation}}
	\def\ee{\end{equation}}
\def\bse{\begin{subequations}}
	\def\ese{\end{subequations}}

\definecolor{deeppurple}{rgb}{0.5, 0, 0.7}

\def\em{\endgroup}

\def\Wr{\mathop{\rm Wr}\nolimits}
\def\gl{\mathrel{\mathpalette\overl@ss>}}
\def\half{{\textstyle\frac12}}

\def\diag{\mathop{\rm diag}\nolimits}

\def\Real{\mathbb{R}}

\def\Complex{\mathbb{C}}

\def\i{\text{i}}

\def\Re{\mathop{\rm Re}\nolimits}
\def\Im{\mathop{\rm Im}\nolimits}
\def\Res{\mathop{\rm Res}\limits}

\def\d{\mathrm{d}}

\def\e{\mathop{\rm e}\nolimits}
\def\@#1{{\mathbf{#1}}}
\def\_#1{{\mathsf{#1}}}
\let\~=\tilde
\let\==\bar
\let\^=\hat
\let\<=\langle
\let\>=\rangle

\def\note[#1]{\marginpar{\color{red}[#1]}}

\def\Xint#1{\mathchoice
	{\XXint\displaystyle\textstyle{#1}}   {\XXint\textstyle\scriptstyle{#1}}   {\XXint\scriptstyle\scriptscriptstyle{#1}}   {\XXint\scriptscriptstyle\scriptscriptstyle{#1}}   \!\int}
\def\XXint#1#2#3{{\setbox0=\hbox{$#1{#2#3}{\int}$}
		\vcenter{\hbox{$#2#3$}}\kern-.5\wd0}}

\def\dashint{\Xint-}
\def\D{{\mathcal D}}

\def\H{{\mathcal H}}

\def\1{{\bf 1}}

\def\z{\zeta}

\def\e{\mathrm{e}}
\def\z{\zeta}
\makeatother

\let\trueparagraph=\paragraph
\def\paragraph#1{\par\smallskip\trueparagraph{\rm\textbf{#1}}}
\advance\textheight 12mm
\advance\voffset -6mm
\allowdisplaybreaks

\begin{document}
\pagestyle{plain}
\title{\bf Local and global well-posedness of the Maxwell-Bloch system of equations with inhomogeneous broadening}
\author{\large
Gino Biondini, Barbara Prinari and Zechuan Zhang}
\textinfo
{Department of Mathematics, State University of New York, Buffalo, NY, 14260}
%\date{\small\today}
\maketitle

\kern-4ex
\begin{abstract}
\noindent
The Maxwell-Bloch system of equations with inhomogeneous broadening is studied,
and the local and global well-posedness of the corresponding initial-boundary value problem is established
by taking advantage of the integrability of the system and making use of the corresponding inverse scattering transform.
A key ingredient in the analysis is the $L^2$-Sobolev bijectivity of the direct and inverse scattering transform
established by Xin Zhou for the focusing Zakharov-Shabat problem.
\end{abstract}

%%%%%%%%%%%%%%%%%%%%%%%%%%%%%%%%%%%%%%%%%%%%%%%%%%%%%%%%%%%%%%%%%%%%%%%%%%%%%%%%%%%%%%%%%%%%%%%%%%%%%%
\section{Introduction and main results}

The Maxwell-Bloch equations (MBEs) describe resonant interaction between light and optical media which
underlies several types of practical devices such as lasers and optical amplifiers. For many experimental
setups, the theoretical description of the interaction between light and an active optical medium is
semi-classical, with the light described classically and the medium quantum-mechanically. Under suitable
physical assumptions (e.g., monochromatic light, one single resonant transition, unidirectional propagation,
etc), averaging over the fast oscillations of the optical pulse yields a description only in terms of the slowly
varying envelopes corresponding to the evolution of the light intensity and phase. Remarkably, even this
simple case produces a host of important physical effects such as electromagnetically induced and self-induced
transparency \cite{McCall1967, mccall69, Lamb1969, PhysRevLett.29.1211, SlusherGibbs, BEMS, PhysRevLett.66.2593, harris:36}, superradiance and superfluorescence \cite{D1954, BL1975, PhysRevLett.36.1035, PSV1979}, photon echo \cite{KAH1964, PS1968, ZM1982},
and even the slowing down of light to a tiny fraction of its speed in vacuum \cite{HDB1999, Slowlight99, FL2000, Milonni2005, RVB2005-1, RVB2005-2}.

In this work, we consider the Cauchy problem for the MBEs, which, in dimensionless form in a comoving reference frame, can be written as
\vspace*{-1ex}
\bse
%\label{e:Cauchy}
\label{e:MBE}
\begin{gather}
\partial_z q(t,z) + \int_\Real P(t,z,k)\,g(k)\,\d k = 0\,,
\label{e:MBE1}\\
\partial_t P(t,z,k) -2\i k\, P(t,z,k)=-2D(t,z,k) q(t,z)\,,
\label{e:MBE2}
\\
\partial_t D(t,z,k)=2 \Re \big[ q^*(t,z) P(t,z,k) \big]\,,
\label{e:MBE3}
\end{gather}
\ese
where the asterisk denotes complex conjugation,
with the following initial-boundary conditions
\bse
\label{e:ICBC}
\vspace*{-1ex}
\begin{gather}
q(t,0)=q_0(t),\quad t\in\Real,\\
\lim_{t\to\pm\infty}q(t,z)=0, \quad z\ge0,\\
\label{e:ICBC2}
D_-(z,k):=\lim_{t\to-\infty}D(t,z,k),\quad z\ge 0,\\
P_-(z,k):=\lim_{t\to-\infty}\e^{-2\i kt}P(t,z,k),\quad z\ge 0.
\label{e:ICBC3}
\end{gather}
\ese
The MBEs~\eqref{e:MBE} describe the propagation of an electromagnetic pulse $q(t,z)$ in a two-level medium characterized by a (real)
population density function $D(t,z,k)$ and a (complex) polarization
fluctuation $P(t,z,k)$  for the atoms \cite{Lamb71}.
Here, $z=z_{\mathrm{lab}}$ is the propagation distance, $t=t_{\mathrm{lab}}-z_{\mathrm{lab}}/c$ is a retarded time ($c$ being the speed of light in vacuum),
%subscripts $z$ and $t$ denote partial derivatives,
the parameter $k$ is the deviation of the transition frequency of the atoms from its mean value.
%the asterisk denotes complex conjugation.
Note that, owing to \eqref{e:MBE2}, $P(t,z,k)$ does not have finite limits as $t\to\pm\infty$ for any $k\ne0$ if $D_-(z,k)\ne0$,
which explains the peculiar form of the boundary condition in Eq.~\eqref{e:ICBC3}.
The precise functional classes to which the initial-boundary data should belong will be clarified below.

The quantity $g(k)$ appearing in \eqref{e:MBE1} is the so-called inhomogeneous broadening function, which accounts for the detuning from the exact quantum
transition frequency due to the Doppler shift caused by the thermal motion of the atoms in the medium.
As such, the function $g(k)$ serves as a density function of a continuous variable $k$, and satisfies the following properties:
\vspace*{-1ex}
\be
g(k)\geq 0,\quad \int_\Real g(k)\,\d k=1.
\ee
A natural choice for the inhomogeneous broadening corresponds to a Lorentzian detuning:
%\marginpar{Do we want to use a letter other than $\epsilon$ to indicate width need not be small?}
\vspace*{-1ex}
\be
\label{e:g}
g(k)=\frac{\epsilon}{\pi(\epsilon^2+k^2)},
\ee
where the parameter $\epsilon>0$ is the detuning width.
The case $g(k)=\delta(k-k_o)$
[with $\delta(\cdot)$ denoting the Dirac delta]
describes the so-called ``sharp-line'' --- or infinitely narrow line --- limit at an arbitrary $k_o\in \Real$
which can be taken to be zero without loss of generality.

It is convenient to introduce a density matrix $\rho(t,z,k)$ that without loss of generality can be assumed to be traceless, i.e., such that
\be
\rho(t,z,k)=\begin{pmatrix} D(t,z,k) & P(t,z,k) \\ P^*(t,z,k) & -D (t,z,k) \end{pmatrix}.
\ee
Moreover, from the MBEs \eqref{e:MBE2} and \eqref{e:MBE3} if follows that $\partial_t \left(D^2(t,z,k)+|P(t,z,k)|^2\right)=0$
for all $k\in \Real$ and all $z\ge 0$, so one can assume, again without loss of generality, that
%and with determinant equal to $-1$ for all $z\ge 0$, so
\be
D^2(t,z,k)+|P(t,z,k)|^2=1\,.
\label{e:DPconstraint}
\ee
As Eqs.~\eqref{e:ICBC} indicate, the medium is assumed to be semi-infinite, i.e., $z\ge 0$, and ``prepared'' in the distant past
(i.e., as $t\rightarrow -\infty$)
in a (known) state characterized by assigned values for the distribution of atoms in the ground and excited states and for the polarization via the asymptotics of $D(t,z,k)$
and $P(t,z,k)$ as $t\to -\infty$ for every $z\ge 0$ and $k\in \Real$. Macroscopically, the medium can
be in either:
(i) a pure ground state, with all atoms in the lowest energy level (i.e., $D_-\equiv -1$ and $P_-\equiv 0$);
(ii) a pure excited state (a medium with a complete ``population inversion'', with all the atoms in the excited state
(i.e., $D_-\equiv 1$ and $P_-\equiv 0$);
(iii) a mixed state with an assigned fraction of atoms in each state ($-1<D_-(z,k)<1$),
in which case the medium exhibits nontrivial polarization fluctuations, encoded by $P_-(z,k)\ne 0$.
An electromagnetic pulse $q(t,0)$ is then injected into the medium at the origin and it propagates into it ($z>0$).

The MBEs~\eqref{e:MBE} can then be written in matrix form as
\bse
\label{e:MBEmatrix}
\begin{gather}
Q_z(t,z) + \half \int_\Real[\sigma_3,\rho(t,z,k)]\,g(k)\,\d k = 0\,,\\
\rho_t(t,z,k) = [ \i k \sigma_3 + Q(t,z),\rho(t,z,k)]\,,
\end{gather}
\ese
where $[A,B] = AB-BA$ is the matrix commutator,
$\sigma_1,\sigma_2,\sigma_3$ are the standard Pauli matrices, with $\sigma_3 = \diag(1,-1)$,
and
\be
Q(t,z)=\begin{pmatrix} 0 & q(t,z) \\ -q^*(t,z) & 0 \end{pmatrix}\,.
\ee
The MBEs~\eqref{e:MBE} are integrable, and their integrability makes it possible to linearize the
initial-boundary value problem (IBVP) (\ref{e:MBE}--\ref{e:ICBC})
via the \textit{inverse scattering transform} (IST)
\cite{Lamb74,Ablowitz74,Zakharov80,ZM1982,M1982,GMZ1983,GMZ1984,GMZ1984b,MN1986,steudel90}.
Specifically, a Lax pair for the MBEs~\eqref{e:MBE} is given by \cite{Ablowitz74}
\vspace*{-0.4ex}
\bse
\label{e:Laxpair}
\begin{gather}
v_t= X v\,,
\label{e:ZS}
\\
v_z=T v\,, \label{e:auxLaxop}
\end{gather}
\ese
with
\vspace*{-1ex}
\bse
\begin{gather}
X(t,z,k) = \i k\sigma_3 +Q,\qquad
T(t,z,k)=\frac{\i\pi}{2}\H_k[\rho(t,z,\zeta)g(\zeta)]\,,
\\
\noalign{\noindent
where $\H_k[f(\zeta)]$ is the Hilbert transform,}
\H_k[f(\zeta)]=\frac{1}{\pi}\dashint_\Real \frac{f(s)}{s-k}\d s\,,
\label{e:Laxpair2}
\end{gather}
\ese
and the symbol $\dashint$ denotes the principal value integral.
(Specifically, Eqs.~\eqref{e:MBEmatrix} are equivalent to the compatibility condition $v_{xt} = v_{tx}$ of~\eqref{e:Laxpair}, i.e., the zero curvature condition $X_t - T_x + [X,T]=0$.)

As usual, the first half of the Lax pair~\eqref{e:Laxpair} is referred to as the scattering problem,
$k$ as the scattering (or spectral) parameter,
and $q(t,z)$ as the scattering potential.
%coincides with the scattering problem for the focusing NLS equation \cite{APT2004,ZS1972}.
%As a result, the formulation of the direct problem is similar to that of the IST for the focusing NLS equation.
%% with one-sided nonzero background \cite{CM2015}.
%On the other hand, as in the case of zero background \cite{Ablowitz74} and symmetric nonzero background \cite{BGKL},
%the ``evolution'' of the scattering data for the MBEs is substantially different and more complicated
%than for the NLS equation, and also from the case of MBEs with a symmetric NZBG.
%
%Moreover, the formulation of the inverse problem in the present work is also substantially more involved than in the focusing NLS case \cite{CM2015}.
Importantly, the scattering problem for the two-level MBEs [namely, \eqref{e:ZS}],
is the celebrated Zakharov-Shabat (ZS) or Ablowitz-Kaup-Newell-Segur (AKNS) system,
which is exactly the same as for the focusing NLS equation \cite{APT2004,ZS1972} (apart from the common switch in the role of the spatial and temporal variables encountered in all signalling problems).
Therefore one can rely on a vast literature for the direct and inverse problems,
both with decaying and non-decaying optical pulses. On the other hand, the propagation along the medium, as well as the coupling with the density matrix $\rho(t, z, k)$, are novel and aspects in the IST for the MBEs~\eqref{e:MBE}.

The IST to solve the initial value problem for the above MBEs with a localized optical pulse [i.e., with $q(t,z)\to0$ as $t\to\pm\infty$]
was first developed in \cite{Ablowitz74} in the case of an initially stable medium
(i.e., in the case $\lim_{t\rightarrow -\infty}D(t,z,k)=-1$) and subsequently generalized to the case of an arbitrary initial state of the medium \cite{Zakharov80,GMZ1983,GMZ1984,GMZ1984b}.
The IST with a symmetric nonzero background (NZBG) [i.e., $q(t,z)\to q_\pm (z)$ with $|q_+(z)|=|q_-(z)|=q_o$ as $t\to \pm \infty$]
was carried out in \cite{BGKL}.
More recently, the IST for the MBEs with  inhomogeneous broadening
and one-sided nonzero background
[i.e., $\lim_{t\to-\infty}q(t,z)=0$ and $\lim_{t\to+\infty}q(t,z)=q_+(z)$ with $|q_+(z)|=A$ for all $z\geq0$] was developed in \cite{MBEonesided}.

Despite the large number of works on the MBEs~\eqref{e:MBE} and its multi-component generalizations,
no well-posedness results are available to the best of our knowledge,
which partly motivated the present work.
Recently, Li and Miller studied the MBEs in the sharp-line limit~\cite{LiMiller}.
Their work raised interesting questions, which also partly motivated the present study.
It is worth mentioning that, even though the MBEs simplify considerably in the sharp-line limit,
this limiting case also restricts the types of physical problems one can describe.
For instance, in the case of decaying optical pulses which we are interested in, the MBEs in the sharp-line limit are only compatible if the medium is initially prepared in a pure (stable or unstable) state.
On the other hand, the presence of inhomogeneous broadening also allows considering a medium initially in a mixed state, without necessarily requiring a compatible non-vanishing optical pulse in the distant past.
As we discuss below, besides its obvious physical relevance, we believe that including inhomogeneous broadening is also crucial to circumvent some of the problems highlighted in \cite{LiMiller}.
Moreover, since the sharp-line case can be recovered as an appropriate (though singular) limit of a generic inhomogeneous broadening function (e.g., $\lim_{\epsilon \to 0} g(k)$ in \eqref{e:g}),
our results should also shed additional light on the case of a narrow line and the limiting sharp-line regime.

The goal of the present work is to establish the local and global well-posedness of the IBVP for the MBEs~\eqref{e:MBE} with initial-boundary conditions~\eqref{e:ICBC} and inhomogeneous broadening in the
case of rapidly decaying initial conditions, as specified by~\eqref{e:ICBC2}.
Besides its intrinsic significance because of the physical relevance of the MBEs,
the importance of the well-posedness result lies in the fact that, in~\cite{LiMiller}, it was shown
that, in the sharp-line limit,
a causality requirement (i.e., $q(t, z) = 0$ for all $t < 0$, $z >0$) must be imposed on both the initial conditions and the solution, otherwise the IBVP for the MBEs in the initially unstable case is ill-posed.
(Specifically, for the same initial datum $q_0(t)$ and $D_-=+1$, the IBVP admits multiple noncausal solutions which decay to both stable and unstable pure states as $t \to + \infty$, see Corollary 4 in~\cite{LiMiller}.)
Conversely, given a causal incident pulse, there exists at most one causal solution to the IBVP for the MBEs in the sharp line limit (Theorem~1 in \cite{LiMiller}).
A different requirement of causality was imposed in \cite{Zakharov80} to ensure uniqueness of solutions of the Gelfand-Levitan-Marchenko (GLM) equations of the inverse problem
(which is related to the uniqueness of solutions of the Riemann-Hilbert problem),
but the MBEs considered in \cite{Zakharov80} were also restricted to the sharp-line case.
On the other hand, it was unclear whether the ill-posedness of the IBVP (or non-uniqueness of solutions of the GLM equations)
extends to the MBEs with inhomogeneous broadening.
In this respect, note that the proof of uniqueness of a causal solution provided in \cite{LiMiller} does not rely on integrability, but rather on a symmetry that
is only valid in the sharp-line limit, and does not extend to the case inhomogeneous broadening.
Regardless, the results of~\cite{LiMiller} do not necessarily imply that, if causality is not imposed, the IBVP for the MBEs with inhomogeneous broadening is also ill-posed.
In fact, in \cite{Zakharov80} Zakharov suggested that the non-uniqueness is related to the so-called ``spontaneous'' solutions (namely, solutions induced by the initial polarization fluctuations of the medium $P_-(z, k)$), and it is due to the behavior of the reflection coefficient of the IST in a small neighborhood of $k=0$; the causality requirement forces the analytic extension of the reflection coefficient at the origin, and allows to recover uniqueness of solution. This conjecture relates the non-uniqueness of solution to the essential singularity of the reflection coefficient at the origin, which is introduced by the sharp-line limit, and we show in this work that when inhomogeneous broadening is included the IBVP for the MBEs is indeed well-posed, under suitable assumptions on the functions $D_-,P_-$ that describe the initial preparation of the medium, and the value of the optical pulse $q(t,0)$ injected in the medium.

Specifically,
letting $D_-(k,z)=\cos{d(k,z)}$ and $P_-(k,z)=e^{ip(k,z)}\,\sin{d(k,z)}$
(which can be done without loss of generality owing to~\eqref{e:DPconstraint}),
in this work we study the IBVP for the system~\eqref{e:MBE} with initial conditions $q(t,0)\in L^{2,1}(\Real)\cap H^1(\Real)$,
and boundary conditions $d(k,z)$ and $p(k,z)$ admitting weak derivatives $d',p'$ with respect to $k$,
and such that $d'(k,z),p'(k,z)\in L^\infty(\Real\times [0,Z])$ for some $Z>0$.
As a special case, the above class includes the physically relevant situation in which $D_-$ and $P_-$
are independent of $z$,
i.e., the case of a spatially uniform medium preparation.
For this class of initial and boundary conditions,
we extend Zhou's $L^2$-Sobolev bijectivity result about the IST for the focusing NLS equation to the MBEs~\eqref{e:MBE} with  inhomogenous broadening.
In turn, this allows us to prove the local and global well-posedness of the problem.
The main results of this work are Theorems~\ref{t:localwellposedness} and~\ref{t:globalwellposedness} in Section~\ref{s:main},
which establish the local and global well-posedness of the system~\eqref{e:MBE} with the given class of initial-boundary conditions.
For concreteness, we will formulate all results in the explicit case of the Lorentzian inhomogeneous broadening function~\eqref{e:g},
but we expect that the results can be generalized to a broad class of detuning functions without major modifications.

The outline of paper is as follows.
In Section~\ref{s:IST} we briefly review the IST for the MBEs with inhomogeneous broadening, %in the case of rapidly decaying optical pulses,
in order to set the notation and introduce relevant quantities that will be used in the rest of the work.
In Section~\ref{s:main} we establish the main result of the paper, namely, the $L^2$-Sobolev bijectivity result for the MBEs with inhomogeneous broadening,
which in turn yields the local and global well-posedness of the system.
In Section~\ref{s:asymp} we discuss the asymptotic states of the medium and of the optical pulse for large $z$,
which establishes appropriate control of the scattering data that is needed to obtain the desired results.
Finally,
Section~\ref{s:concl} is devoted to some concluding remarks.

%%%%%%%%%%%%%%%%%%%%%%%%%%%%%%%%%%%%%%%%%%%%%%%%%%%%%%%%%%%%%%%%%%%%%%%%%%%%%%%%%%%%%%%%%%%%%%%%%%%%%%
\section{Overview of the IST}
\label{s:IST}

In this section we concisely review existing results on the direct and inverse scattering problem for the MBEs~\eqref{e:MBE}, and on the propagation in $z$ of eigenfunctions and scattering data.
%to set the notation and introduce relevant quantities that will be used in the rest of the work.
Since the scattering problem \eqref{e:ZS} coincides with the Zakharov-Shabat problem for the focusing NLS equation,
the results of section~\ref{s:direct} and section~\ref{s:inverse} are well known,
and we therefore omit all proofs.
For further details, we refer the reader to the many standard references on the subject, such as~\cite{APT2004,ZS1972,BealsCoifman,Zhou1989,Zhou1998}.

\subsection{Direct scattering problem}
\label{s:direct}

The direct problem in the IST consists in constructing a map from the solution of the MBEs $q(t,z)$ at a fixed $z\ge 0$ into a suitable set of scattering data.
As usual, this is done by introducing two sets of Jost eigenfunctions,
which are solutions of the scattering problem with prescribed exponential asymptotic behavior as $t\to \pm \infty$, respectively,
as well as scattering data that relates the two sets of Jost eigenfunctions.
The analyticity properties of eigenfunctions and scattering data as functions of the spectral parameter $k\in \Complex$,
and their asymptotic behavior as $k\to \infty$ are crucial to set up the inverse problem.

\paragraph{Jost solutions, analyticity and scattering matrix.}

In light of the asymptotic behaviors of the scattering problem, namely, Eq.~\eqref{e:ZS}, as $|t|\to\infty$,
we define the Jost eigenfunctions as
\vspace*{-1ex}
\be
\label{e:Jostsols}
\phi_\pm(t,z;k) = \e^{\i kt\sigma_3}\,(1 + o(1)), \qquad t\to \pm\infty\,,
%\phi_+(t,z,k) = \e^{\i kt\sigma_3}\,(1 + o(1)),  \qquad t\to +\infty\,.
\ee
and introduce modified eigenfunctions by removing the asymptotic exponential oscillations from the Jost solutions,
i.e.,
$\mu_\pm(t,z;k)=\phi_\pm(t,z;k)\e^{-\i kt\sigma_3}$, so that
$\mu_\pm(t,z;k)=I+o(1)$ as $t\to \pm \infty$.
%Then the first part of the Lax pair \eqref{e:Laxpair} becomes
%\be\label{e:Laxmut}
%\mu_t-\i k[\sigma_3,\mu]=Q\mu\,,
%\ee
%Moreover, the asymptotic behavior of  the solutions is:
%\be
%\mu_\pm(t,z;k)=I+o(1)$,\qquad t\to\pm\infty.
%\ee
%Therefore, integrating \eqref{e:Laxmut} gives the integral equations for the Jost solutions:
The modified eigenfunctions are uniquely defined by the following integral equations:
\vspace*{-1ex}
\begin{gather}
\label{e:int_eq_efs}
%\mu_-(t,z;k)=I+\int_{-\infty}^t\e^{\i k\sigma_3(t-\tau)}Q(\tau,z)\mu_-(\tau,z;k)\e^{-\i k\sigma_3(t-\tau)}\d\tau,\\
\mu_\pm(t,z;k)=I\mp\int^{\mp\infty}_t\e^{\i k\sigma_3(t-\tau)}Q(\tau,z)\mu_\pm(\tau,z;k)\e^{-\i k\sigma_3(t-\tau)}\d\tau.
\end{gather}
One can then show that the vector eigenfunctions can be analytically extended in the complex $k$-plane into the following regions: $\mu_{-1}$ and $\mu_{+2}$ are analytic in the lower half plane (LHP, $\Im k<0$), whereas $\mu_{+1}$ and $\mu_{-2}$ are analytic in the upper half plane (UHP, $\Im k>0$), where $\mu_{\pm j}$ is the $j$th column of the matrix $\mu_{\pm}$.
The analyticity properties of the columns of $\phi_\pm$ follow trivially from those of $\mu_{\pm}$:
\vspace*{-1ex}
\bse
\label{e:analyphi}
\begin{gather}
	\phi_{-1},\phi_{+2}:\quad \rm LHP\,,\qquad
	\phi_{+1},\phi_{-2}:\quad \rm UHP\,.
\end{gather}
\ese
As usual, by Abel's theorem we know that
$\partial_t(\det v)=0$ for any matrix solution $v$ of Eqs.~\eqref{e:Laxpair}.
In addition, for all $z\in\Real$, $\lim_{t\to \pm\infty}\phi_\pm(t,z;k)= \e^{\i kt\sigma_3}$.
Hence,
$\forall t\in \Real$ we have
$\det \phi_\pm(t,z,k)=1$.
Thus, for all $k\in \Real$, both
$\phi_-$ and $\phi_+$ are two fundamental matrix solutions of the scattering problem,
and one can express one set of eigenfunctions in terms of the other one:
\vspace*{-1ex}
\be
\label{e:scattmatrix}
\phi_+(t,z;k)=\phi_-(t,z;k)S(k,z),\qquad k\in\Real,
\ee
where $S(k,z)$ is the scattering matrix, whose entries are referred to as the scattering coefficients.
The scattering matrix is unimodular, since~\eqref{e:scattmatrix} implies $\det S = \det\phi_\pm = 1$.
As usual, if we write $S(k,z)=(s_{ij})_{2\times 2}$, the scattering coefficients $s_{ij}$ can be expressed as Wronskians of the Jost solutions:
\vspace*{-1ex}
\begin{subequations}
\label{Wrapz}
\begin{align}
s_{11}(k,z)=\Wr(\phi_{+1}(t,z;k),\phi_{-2}(t,z;k)), \qquad
s_{12}(k,z)=\Wr(\phi_{+2}(t,z;k),\phi_{-2}(t,z;k)), \\
s_{21}(k,z)=\Wr(\phi_{-1}(t,z;k),\phi_{+1}(t,z;k)), \qquad s_{22}(k,z)=\Wr(\phi_{-1}(t,z;k),\phi_{+2}(t,z;k)).
\end{align}
\end{subequations}
Combining the Wronskians \eqref{Wrapz} with the analyticity of the eigenfunctions \eqref{e:analyphi}, we have the analyticity of the diagonal components of the scattering matrix:
\be
s_{11}(k,z): {\rm UHP}\,,\qquad s_{22}(k,z): {\rm LHP}\,,
\ee
while in general the off-diagonal entries $s_{12}$ and $s_{21}$ are only defined for $k\in \Real$ and do not admit analytic continuation in the complex $k$-plane.
Finally, the reflection coefficients are defined as:
\vspace*{-1ex}
\be
\label{e:refcoeff}
r(k,z)=\frac{s_{21}(k,z)}{s_{11}(k,z)},\quad \tilde{r}(k,z)=\frac{s_{12}(k,z)}{s_{22}(k,z)},\quad k\in\Real.
\ee

\paragraph{Symmetries of eigenfunctions and scattering coefficients.}
%\label{s:symmetries}

We begin by discussing the symmetries of the eigenfunctions.
Note that $Q^\dag(t,z)=-Q(t,z)$, where $^\dag$ denotes conjugate transpose. Letting $w(t,z;k)=(\phi^\dag(t,z;k))^{-1}$, it is easy to show that if $\phi(t,z;k)$ is a solution of the scattering problem, so is $w(t,z;k)$. Now let us restrict our attention to the real $k$ axis. Taking $\phi=\phi_\pm$ we see that the asymptotic behavior of $w$ as $t\to \pm \infty$ coincides with that of $\phi_\pm$. Because the solution of the scattering problem with given BC is unique, we have
\vspace*{-1ex}
\be
\label{e:sym1}
\phi_\pm^{-1}=\phi_\pm^\dag,\quad k\in\Real,
\ee
which is equivalent to
\vspace*{-1ex}
\be
\label{e:sym1'}
\phi^*_\pm(t,z;k)=\sigma_2\phi_\pm(t,z;k)\sigma_2,\quad k\in\Real, \qquad
\sigma_2=\begin{pmatrix}0 & -\i \\ \i & 0 \end{pmatrix}\,.
\ee
%where $\sigma_*=\i\sigma_2$.

Now we discuss the resulting symmetries of the scattering matrix and scattering coefficients.
From \eqref{e:sym1} and the scattering relation \eqref{e:scattmatrix}, it follows that the scattering matrix $S(k,z)$ satisfies
$S^{-1}(k,z)=S^\dag(k,z)$ for $k\in\Real$, i.e.
\vspace*{-1ex}
\bse
\label{e:symsij}
\begin{gather}
s_{11}(k,z)=s_{22}^*(k^*,z),\quad \Im k\geq 0,
\label{e:symSd}
\\
s_{12}(k,z)=-s_{21}^*(k,z),\quad k\in\Real.
\label{e:symSo}
\end{gather}
\ese
Moreover, recalling $\det S(k,z)=1$ and the symmetries \eqref{e:symsij}, we obtain the following identity:
\vspace*{-1ex}
\be
\label{e:detS}
|s_{11}(k,z)|^2+|s_{21}(k,z)|^2=1,\quad k\in\Real.
\ee
Combining \eqref{e:symSd} and \eqref{e:symSo} we obtain the the symmetry between reflection coefficients:
\vspace*{-1ex}
\be
\label{e:symr}
\tilde{r}(k,z)=-r^*(k,z),\quad k\in\Real.
\ee
Again, because the reflection coefficients contain the off-diagonal entries of the scattering matrix, in general, $r$ and $\~r$  cannot be extended off the real $k$-axis.

\paragraph{Asymptotic behavior as $k\to\infty$.}
%\label{sec:asymptzvar}

The asymptotic properties of the eigenfunctions and the scattering matrix are instrumental in properly normalizing the inverse problem. Moreover, the next-to-leading-order behavior of the eigenfunctions will allow us to reconstruct the potential $q(t,z)$ from eigenfunctions. Here we summarize the results, which can
be obtained by integration by parts on the integral equations \eqref{e:int_eq_efs}:
%\vspace*{-1ex}
\be
\label{e:mu_largek}
\mu_\pm(t,z;k)=I\pm \frac{1}{2\i k}\sigma_3 Q(t,z)\pm \frac{\sigma_3}{2\i k}\int_{\pm \infty}^t|q(\tau,z)|^2\d\tau+O(k^{-2})\,,
\qquad
k\to\infty\,,
\ee
with the expansion valid for $k\in\Real$ as well as in the region of analyticity of each column.
In turn, this gives
\vspace*{-1ex}
%\bse
\begin{gather}
%\mu_\pm(t,z;k)=I+O(k^{-1}),
%\\
S(k,z)=I+O(k^{-1}),
\qquad
k\to\infty\,,
\label{e:largek}
\end{gather}
%\ese
again with the expansion valid for $k\in\Real$ as well as in the region of analyticity of each entry.
In particular, the latter equation shows that for any fixed $z\ge 0$ the analytic scattering coefficients
$s_{11}(k,z)$ and $s_{22}(k,z)$ cannot vanish as $k\to \infty$ in the appropriate half-plane.

\paragraph{Discrete eigenvalues and residue conditions.}
~\kern-0em
The zeros of $s_{11}(k,z)$ and $s_{22}(k,z)$ comprise the discrete eigenvalues of the scattering problem in \eqref{e:Laxpair}.
Since for any fixed $z\ge 0$ $s_{11}(k,z)$ is analytic in the UHP of $k$, it has at most a countable number of zeros there.
Moreover, owing to the symmetry \eqref{e:symSd}, $s_{11}(k,z)=0$ if and only if $s_{22}(k^*,z)=0$.
That is, discrete eigenvalues appear in complex conjugate pairs.

For the ZS spectral problem, discrete eigenvalues can be located anywhere in the complex $k$-plane, they are not necessarily simple, and one cannot a priori exclude the existence of zeros of $s_{11}(k,z)$ and $s_{22}(k,z)$ for $k\in \Real$.
%(the latter are often referred to as spectral singularities \cite{Zhou1989}, or embedded eigenvalues).
%
Any discrete eigenvalue that lies on the real axis is called a \textit{spectral singularity},
as is any accumulation point of discrete eigenvalues \cite{Zhou1989}.
These singularities correspond to resonant states or bound states in the system.
The analysis of spectral singularities is a key aspect of understanding the scattering behavior and the spectrum of the scattering problem.
For rapidly decaying potentials, there exists a characterization of the location of (real) spectral singularities, as well as sufficient conditions on $q(t, z)$ that guarantee their absence (for instance, spectral singularities are absent in the case of single lobe potentials, and certain double and multiple lobe potentials \cite{KlausShaw03, CM2005v379p21, JMP2007v48p123502}). On the other hand, there are potentials in the Schwartz class for which discrete eigenvalues can accumulate to spectral singularities, and spectral singularities themselves can accumulate on the continuous spectrum \cite{Zhou1989}.
%(in which case, $N=\infty$).
%
However, %Zhou's approch in \cite{Zhou1998} to derive $L^2$ bijectivity for ZS scattering problem
the IST can still be effectively formulated even in such cases
%if spectral singularities and/or an infinite number of discrete eigenvalues are present
\cite{Zhou1989}
(e.g., see Remark~\ref{r:ss} in section~\ref{s:inverse}).

For simplicity and concreteness, in the following paragraphs we discuss explicitly the case in which there is a finite number $N$ of discrete eigenvalues
(i.e., zeros of $s_{11}(k,z)$), and all such zeros are simple.
That is,
$s_{11}(k_n,z)=0$ and $s_{11}'(k_n,z)\neq0$ with $\Im k_n>0$ for $n=1, 2, \dots, N$,
where hereafter the prime will denote differentiation with respect to $k$.
The possible presence of higher-order zeros introduces some minor technical complications, but no conceptual differences.
Moreover, the possible presence of spectral singularities
and that of an infinite number of zeros can be dealt with using Zhou's approach~\cite{Zhou1989}.
Therefore, as we discuss in section~\ref{s:inverse},
the results of this work
also apply in the presence of an arbitrary (possibly infinite) number of zeros of arbitrary multiplicity as
well as in the presence of arbitrary spectral singularities.
%are not limited to the case of a finite number of simple zeros off the real axis. \eb

For all $n=1,\dots,N$,
at $k=k_n$ we have  $\Wr(\phi_{+1}(t,z;k_n),\phi_{-2}(t,z;k_n))=0$  from \eqref{Wrapz}. Thus, there exists a scalar  function $c_n(z)\neq0$  so that  $\phi_{+1}(t,z;k_n)=c_n(z)\phi_{-2}(t,z;k_n)$.
Similarly, at $k=k_n^*$ we have  $\Wr(\phi_{-1}(t,z;k_n^*),\phi_{+2}(t,z;k_n^*))=0$,  which implies that %there exists a $\tilde{c}_n(z)\neq0$so that
 $\phi_{+2}(t,z;k_n^*)=\tilde{c}_n(z)\phi_{-1}(t,z;k_n^*)$.
Thus we have the following residue conditions:
\vspace*{-1ex}
\bse
\label{residue}
\begin{gather}
	\Res_{k=k_n}\left[\frac{\mu_{+1}(t,z;k)}{s_{11}(k,z)}\right]= C_n(z) \e^{-2\i k_nt}\mu_{-2}(t,z;k_n)\,,\\
	\Res_{k=k_n^*}\left[\frac{\mu_{+2}(t,z;k)}{s_{22}(k,z)}\right]= \tilde{C}_n(z) \e^{2\i k_n^*t}\mu_{-1}(t,z;k_n^*)\,,
\end{gather}
\ese
where
\be
\label{e:Cn}
 C_n(z) =\frac{c_n(z)}{s'_{11}(k_n,z)},\quad  \tilde{C}_n(z) =\frac{\tilde{c}_n(z)}{s'_{22}(k_n^*,z)}\,.
\ee
We also have symmetries for the norming constants:
 $\tilde{c}_n(z)=-c_n^*(z)$
for $n=1,\dots,N$, which can be easily derived from the symmetry \eqref{e:sym1'} for the eigenfunctions evaluated at $k=k_n$, and
which imply:
\be
 C_n(z)=-\tilde{C}_n^*(z), \quad n=1,\dots,N.
\ee
Furthermore, at the spectral singularities, since $C_n(z)=\tilde{C}_n(z)$ and $C_n(z)=-\tilde{C}_n^*(z)$, it is apparent that $C_n(z)$ is purely imaginary.

\paragraph{Trace formula.}

One can also obtain ``trace formulae'' to recover the analytic scattering coefficients in terms
of scattering data. In particular, in the case of a finite number of simple discrete eigenvalues, the coefficient $s_{11}$ is given by
\be
\label{e:trace}
s_{11}(k,z)=\prod_{n=1}^N\frac{k-k_n}{k-k_n^*}\exp{\left[-\frac{1}{2\pi i}\int_\Real\frac{\log(1+|r(s,z)|^2)}{s-k}\d s\right]}\,, \quad \Im k>0\,,
\ee
and $s_{22}(k,z)=s_{11}^*(k^*,z)$.

\paragraph{Boundary conditions for the density matrix.}

Generally, the density matrix $\rho(t,z;k)$ does not have a finite limit as $t\to\pm\infty$.
Nonetheless, we next show that one can define proper asymptotic data.
By direct substitution, if $\phi(t,z;k)$ is any fundamental matrix solution of the scattering problem, one has
\be
\frac{\partial}{\partial t}\left(\phi^{-1}(t,z;k)\rho(t,z;k)\phi(t,z;k)\right)=0.
\ee
Therefore, we can define $\rho_\pm(k,z)$ as
\be\label{e:defrhopm}
\rho_\pm(k,z)=\phi^{-1}_\pm(t,z;k)\rho(t,z;k)\phi_\pm(t,z;k).
\ee
Considering the asymptotic behaviors \eqref{e:Jostsols} of $\phi_\pm$, we can obtain from \eqref{e:defrhopm} the following asymptotics:
\be
\rho(t,z;k)=\e^{\i kt\sigma_3} (\rho_\pm(k,z)+o(1)) \,\e^{-\i kt\sigma_3},\quad t\to\pm\infty.
\ee
%Then we can parameterize $\rho_\pm$ as
Letting
\be
\rho_\pm(k,z)=\begin{pmatrix}
	D_\pm&P_\pm\\
	P_\pm^*&-D_\pm
	\end{pmatrix},
\ee
where $D_\pm\in\Real$ and $D_\pm^2+|P_\pm|^2=1$,
%Equation \eqref{e:defrhopm} implies that
%\be
%D_\pm(z,k)=\lim_{t\to\pm\infty}D(t,z,k)\quad P_\pm(z,k)=\lim_{t\to\pm\infty}P(t,z,k).
%\ee
%Conversely,
Eq.~\eqref{e:defrhopm} implies
\be
\rho(t,z;k) = \phi_\pm(t,z;k) \,\rho_\pm(k,z)\, \phi_\pm^{-1}(t,z;k)\,.
\ee
Consequently, we have
\be
\rho(t,z;k)=\e^{\i kt\sigma_3}\rho_\pm(k,z)\, e^{-\i kt\sigma_3}+o(1)=\begin{pmatrix}
	D_\pm&\e^{2\i kt\sigma_3}P_\pm\\
	\e^{-2\i kt\sigma_3}P_\pm^*&-D_\pm
	\end{pmatrix}+o(1),
\qquad t\to\pm\infty.
\ee
Next, we show that $\rho_\pm$ are not independent of each other, so only one of them can be given in order to retain compatibility. It is trivial to see by direct calculation that $S\rho_+=\rho_-S=\phi_-^{-1}\rho\phi_+$, where $S$ is the scattering matrix. Hence,
\be\label{e:relarhopm}
\rho_+(k,z)=S^{-1}(k,z)\rho_-(k,z)S(k,z),\quad k\in\Real.
\ee
Expanding both sides of \eqref{e:relarhopm}, we have
\bse
\begin{gather}
	D_+=(|s_{11}|^2-|s_{21}|^2)D_-+s_{11}s_{21}^*P_-^*+s_{11}^*s_{21}P_-,\\
	P_+=(s_{11}^*s_{12}-s_{21}^*s_{22})D_-+s_{21}^*s_{12}P_-^*+s_{11}^*s_{22}P_-.
\end{gather}
\ese
Using \eqref{e:symsij} and \eqref{e:detS}, we can get the explicit expressions of $D_+$ and $P_+$ in terms of scattering data, namely:
\be\label{e:D+P+}
D_+=\frac{1}{1+|r|^2}((1-|r|^2)D_-+2\Re (rP_-))\,,\quad
P_+=\frac{\e^{-2\i\arg s_{11}}}{1+|r|^2}(P_--(r^2P_-)^*+2r^*D_-).
\ee
Eq.~\eqref{e:D+P+} indicates that even if the medium is
initially prepared in a pure state, i.e., with $P_-=0$ and $D_-=\pm1$, in general one has $P_+\neq0$, unless  $r(k,z)\equiv 0$  for all $k\in \Real$
(i.e., a reflectionless/purely solitonic optical pulse). Note that $\arg s_{11}(k)$, which contributes to the phase of $P_+$, can also be written
in terms of discrete eigenvalues and reflection coefficient using the trace formula \eqref{e:trace}.

%%%%%%%%%%%%%%%%%%%%%%%%%%%%%%%%%%%%%%%%%%%%%%%%%%%%%%%%%%%%%%%%%%%%%%%%%%%%%%%%%%%%%%%%%%%%%%%%%%%%%%
\subsection{Propagation}

The evolution of the scattering data is the part of the IST for which the MBEs~\eqref{e:MBE} differs most significantly
from that for the NLS equation.
On the other hand, even this part of the IST for the MBEs~\eqref{e:MBE} is standard,
and has been discussed in several works.
For further details on the results of this subsection, see \cite{Ablowitz74,BGK2003,CBA2014,Zakharov80,ZM1982,M1982}.%

\paragraph{Simultaneous solutions of the Lax pair and auxiliary matrices.}
Since the asymptotic behavior of the Jost solutions $\phi_\pm$ as $t\to\pm\infty$ is fixed by Eq.~\eqref{e:Jostsols}, in general they will not be solutions of the second half of the Lax pair, namely \eqref{e:auxLaxop}.
We next introduce auxiliary matrices $R_\pm$ that relate the solutions $\phi_\pm$ to a simultaneous fundamental matrix solution of both parts of the Lax pair. Then, we will use these matrices $R_\pm$ to compute the propagation of scattering data. Because both $\phi_+$ and $\phi_-$ are fundamental matrix solutions of the scattering problem, any other solution $\Phi(t,z;k)$ can be written as
\be
\Phi(t,z;k)=\phi_+(t,z;k)A_+(k,z)=\phi_-(t,z;k)A_-(k,z),\quad k\in\Real,
\ee
where $A_\pm(k,z)$ are $2\times2$ matrices independent of $t$ and satisfy the following differential equation:
\vspace*{-1ex}
\be
\partial_z A_{\pm}=\frac{\i}{2}R_\pm A_\pm,\quad k\in\Real,
\label{e:Apm}
\ee
%where the subscript $z$ denotes partial differentiation with respect to $z$,
where
\be
\frac{\i}{2}R_\pm=\phi_\pm^{-1}(V\phi_\pm-\partial_z\phi_{\pm}).
\ee
Consequently, we can express the auxiliary matrices $R_\pm$ in terms of known quantities for all $k\in\Real$:
\bse
\begin{gather}
	R_{\pm,\rm{d}}=\pi\H_k[\rho_{\pm,d}(z;k')g(k')],\quad k\in\Real,\\
	R_{\pm,\rm{o}}=\pm\i\pi g(k)\sigma_3\rho_{\pm,o}(k,z),\quad k\in\Real,
\end{gather}
\ese
where subscripts ''d'' and ''o'' denote the diagonal and off-diagonal parts of a matrix, respectively. Entry-wise, $R_\pm$ is given by
\bse
\begin{gather}
	R_{\pm,11}=\pi\H_k[g(k')]D_\pm,\quad R_{\pm,22}=-R_{\pm,11},\\
	R_{\pm,12}=\pm\i\pi g(k)P_\pm,\quad R_{\pm,21}=\mp\i\pi g(k)P_\pm^*.
\end{gather}
\ese
One can express the scattering matrix, defined in \eqref{e:scattmatrix}, using \eqref{e:Apm} as
\be
S(k,z)=A_-(k,z)A_+^{-1}(k,z).
\ee
With some algebra, one can obtain the ODE satisfied by $S(k,z)$:
\be
\partial_z S=\frac{\i}{2}(R_-S-SR_+),
\ee
from which it follow that for all $k$ in the UHP:
\be
\partial_z s_{11}=\frac{\i}{2}(R_{-,11} s_{11}-R_{+,11} s_{11}),
\label{e:odes11}
\ee
where we used that
\be
R_{+,12}(k,z)=R_{-,21}(k,z)=0 \qquad \Im k>0.\label{e:R12=0}
\ee
The solution of \eqref{e:odes11} is given by
\be
\label{e:s11z}
s_{11}(k,z)=\exp\left\{\frac{\i}{2}\frac{k^2}{\epsilon^2+k^2}\int_{0}^z(D_-(k,\zeta)-D_+(k,\zeta))\d \zeta\right\} s_{11}(k,0),
\ee
proving, in particular, that the discrete eigenvalues, as zeros of $s_{11}(k,z)$, are independent of $z$.

It is also worth noticing that for any fixed $z\ge 0$, the first of Eqs.~\eqref{e:D+P+} shows that
$D_+(k,z)$ has the same behavior as $D_-(k,z)$ as $|k|\to \infty$ since the reflection coefficient
$r(k,z)=s_{21}(k,z)/s_{11}(k,z)\to 0$ as $|k|\to \infty$, consistently with \eqref{e:largek} and \eqref{e:s11z}.

\paragraph{Propagation equations for the reflection coefficient and norming constants.}
Using the auxiliary matrices~$R_\pm$, we can obtain the propagation equation for the reflection coefficient:
\be\label{e:propr}
\frac{\partial r(k,z)}{\partial z}=\pi D_-(k,z)\,\big[g(k)-\i\H_k[g(k)]\big]\,r(k,z)-\pi g(k)P_-^*(k,z).
\ee
Recalling the Hilbert transform of $g(k)$, namely,
\be
\H_k[g(k)]=\frac{k}{\pi (\epsilon^2+k^2)},
\ee
\eqref{e:propr} becomes
%\unskip\footnote{Note that (1.6.8) in Sitai's Ph.D.\ thesis has a mistake, since $\H_k[g(k)] = (k/\epsilon)\,g(k)$, not $g(k)$.}
\be
\frac{\partial r(k,z)}{\partial z} + w(k)D_-(k,z)r(k,z) = -\pi g(k)P_-^*(k,z),
\ee
where
\be\label{e:w}
w(k) = \bigg(\frac{\i k}{\epsilon}-1 \bigg)\,\pi g(k)\equiv -\frac{1}{\i k+\epsilon}\,.
\ee
The ODE~\eqref{e:propr} is easily solved to give
\be
r(k,z) = \e^{-w(k)\D(k,z)} \left( r(k,0) - \pi g(k) \int_{0}^z\e^{w(k)\D(k,\z)}P_-^*(k,\z)\,\d \z \right),
\label{e:rsolution}
\ee
where
\be
\D(k,z) = \int_0^z D_-(k,\z)\,\d \z.
\ee
The second reflection coefficient $\tilde{r}(k,z)$ can be be obtained via the symmetry~\eqref{e:symr}.
The solution~\eqref{e:rsolution} is particularly simple in the case of a medium in a pure state, i.e., $P_- = 0$ and $D_- = \pm1$.
More generally, if the medium is not in a pure state but $P_-$ and $D_-$ are independent of $z$,
\eqref{e:rsolution} yields
\be
r(k,z) = \begin{cases}
    \e^{-w(k)D_- z}\,r(k,0)\,,\qquad &D_- = \pm1\,,\\
    \displaystyle
    \e^{- w(k)D_-(k)z} \left( r(k,0) - \pi \frac{g(k)}{w(k)D_-(k)} P_-^*(k) (\e^{w(k)D_-(k)z}-1)  \right)\,, &-1 < D_-(k)<0~\vee~0<D_-(k)<1\,,\\
        r(k,0) - \pi g(k) P_-^* z \,,\qquad &D_- = 0\,,
    \end{cases}
\label{e:rkz}
\ee
with $w(k)$ still given by~\eqref{e:w}.

Finally, recall the norming constant $C_n$ is given by \eqref{e:Cn}. Thanks to the auxiliary matrices $R_\pm$, one can derive the propagation equation for $C_n$:
\be
\frac{\partial C_n}{\partial z}=-\i R_{-,11}(z,k_n)C_n,\quad n=1,\dots, N,
\ee
where $k_n$ is the corresponding discrete eigenvalue, and $R_{-,11}(k,z)$ is the $(1,1)$-entry of the auxiliary matrix $R_-(k,z)$.

%%%%%%%%%%%%%%%%%%%%%%%%%%%%%%%%%%%%%%%%%%%%%%%%%%%%%%%%%%%%%%%%%%%%%%%%%%%%%%%%%%%%%%%%%%%%%%%%%%%%%%
\subsection{Inverse problem}
\label{s:inverse}

In this section we briefly
discuss the inverse problem in the IST, namely the reconstruction of the solution of the Maxwell-Bloch system~\eqref{e:MBE}
from the knowledge of the scattering data.
We formulate the inverse problem in terms of a matrix Riemann-Hilbert problem (RHP)
for a suitable sectionally meromorphic function in $\Complex\setminus\Real$, with assigned jumps across $\Real$, and then
reconstruct the solution of the MBEs \eqref{e:MBE} from the large-$k$ behavior of the solution of the RHP.
(Note that, while this formulation of the inverse problem is essentially the same as that for the focusing NLS equation
\cite{APT2004},
in the original works the inverse problem for the MBEs was formulated through Gelfand-Levitan-Marchenko equations
\cite{Lamb74,Ablowitz74,Zakharov80,ZM1982,M1982,GMZ1983,GMZ1984,GMZ1984b,MN1986,steudel90}).

\paragraph{The Riemann-Hilbert problem.}
We begin by introducing the following meromorphic matrix-valued function $M(t,z;k)$
based on the analyticity properties of the Jost eigenfunctions and scattering coefficients
discussed in section~\ref{s:direct}:
\be
\label{e:M}
\displaystyle
M(t,z;k)=
\begin{cases}
      \left(
      \dfrac{\mu_{+1}}{s_{11}}\,,\mu_{-2}
      \right),
      & k\in\Complex^+\,,
      \\[2ex]
      \left( \mu_{-1}\,,\dfrac{\mu_{+2}}{s_{22}}\right),
      & k\in \Complex^-\,.
\end{cases}
\ee
It is easy to show that $M(t,z;k)$ satisfies the following RHP:
%\begin{RHP}
%	Find a $2\times2$ matrix-valued function $M(t,z;k)$  such that
	\begin{enumerate}
		\item
	    $M(t,z;k)$ is meromorphic for  $k\in\Complex\setminus\Real$.
	    \item
		$M^\pm(t,z;k)$  satisfy the following asymptotic behaviors as $k\to\infty$:
		\be
		M^\pm(t,z;k)=I+O(k^{-1}),\quad k\to\infty.
		\ee
		\item
		$M(t,z;k)$  satisfies the jump condition
		\be
		M^+(t,z;k)=M^-(t,z;k)G(t,z;k),\quad k\in\Real,
		\ee
		where the jump matrix is
		\be
		G(t,z;k)=\begin{pmatrix}
			            1-|r(k,z)|^2 &  -\e^{2\i kt}r^*(k,z) \\
                        \e^{-2\i kt}r(k,z) & 1
        \end{pmatrix}.
		\ee
		\item
		$M(t,z;k)$ has simple poles at $k=k_n$ and $k=k_n^*$, with the following residue conditions:
		\bse
		\begin{gather}
		\Res_{k=k_n}M^+_1(t,z;k)=C_n\e^{-2\i k_nt}M_2^+(t,z;k_n),\\
		\Res_{k=k_n^*}M^-_2(t,z;k)=\tilde{C}_n\e^{2\i k_n^*t}M_1^-(t,z;k_n^*),
        \end{gather}
		\ese
		where $C_n$ and $\tilde{C}_n$ are defined in \eqref{e:Cn}, and $C_n=\tilde{C}_n=-C_n^*$ if $k_n\in \Real$ (spectral singularity).
	\end{enumerate}

\begin{remark}
In light of \eqref{e:mu_largek},
once the solution of the RHP is known, the solution $q(t,z)$ of the Cauchy problem is recovered as
 \be
q(t,z)=-2\i\lim_{k\to\infty} k\,M_{12}(t,z;k).
\ee
\end{remark}

\begin{remark}
For a pure state (for which $P_-\equiv 0$), %we choose $D_-=\pm1$, the jump matrix becomes
	%\marginpar{[double-check, correct and express entirely in terms of $r(0,k)$]}
	\be
		G(t,z;k)=\begin{pmatrix}
		1 - \e^{2\pi g(k)D_-z}|r(k,0)|^2 & -\e^{2\i\left( kt+ \frac{k}{2\epsilon}\pi g(k)D_-z \right)+\pi g(k)D_-z}r^*(k,0)\\
		\e^{-2\i \left(kt+ \frac{k}{2\epsilon} \pi g(k)D_-z\right)+\pi g(k)D_-z}r(k,0) & 1
	\end{pmatrix}
	\ee
with $D_-=\pm1$.
\end{remark}

\begin{remark}
If spectral singularities are present, there are also poles on the jump contour.
However, they can be dealt with as in \cite{BilmanMiller,SAPM2021v146p371,Zhou1998}.
Specifically, since $s_{11}$ and $s_{22}$ tend to $1$ as $k\to\infty$ in the UHP/LHP, they cannot have zeros in a neighborhood of $k=\infty$.
Therefore, we can introduce a circle $C_\infty$ (oriented counterclockwise) centered at $k=\infty$ such that $M(t,z;k)$
has no singularities
outside $C_\infty$.  This circle,
together with the portion of the real axis $\Real$ outside $C_\infty$, separates the complex plane into
three disjoint regions.
Inside $C_\infty$, one then replaces $M$ with a different matrix that has no singularities, thereby obtaining
a modified RHP without poles or singularities on the jump contour.
See \cite{BilmanMiller,SAPM2021v146p371,Zhou1998} for further details.
The same approach can also be used to deal with higher-order poles, as well as an infinite number of discrete eigenvalues.
Therefore, the results of this work
also apply in the presence of an arbitrary (possibly infinite) number of zeros of arbitrary multiplicity, as
well as in the presence of arbitrary spectral singularities.
\label{r:ss}
\end{remark}

%%%%%%%%%%%%%%%%%%%%%%%%%%%%%%%%%%%%%%%%%%%%%%%%%%%%%%%%%%%%%%%%%%%%%%%%%%%%%%%%%%%%%%%%%%%%%%%%%%%%%%
\section{Bijectivity of the IST and well-posedness of the MBEs}
\label{s:main}

In this section we extend Zhou's $L^2$-Sobolev bijectivity result for the NLS equation \cite{Zhou1998}
to the direct and inverse scattering transform for the Lax system \eqref{e:Laxpair},
and we use the corresponding results to prove the local and global well-posedness of the MBEs~\eqref{e:MBE}.
%We show that the propagation of the scattering data is well defined if the initial condition $q_0\in H^1(\Real)\cap L^{2,1}(\Real)$,
%which also leads to a global existence result for the MBE.
%
Let us first introduce some notations that will be used in this section:
\begin{itemize}
	\item
Since $D_-^2(k,z)+|P_-(k,z)|^2=1$ for all $k\in \Real$ and all $z\ge0$, without loss of generality we can set
\be
 D_-(k,z)=\cos d(k,z),\quad P_-(k,z)=\e^{\i p(k,z)}\,\sin d(k,z),
\ee
where $d(z,k)$ and $p(z,k)$ are real-valued functions.
%\item
The results of this section require $d(k,z)$ and $p(k,z)$ to be weakly differentiable with respect to $k$, and for such derivatives
to be uniformly bounded as functions of $k$ and $z$.
%$W^{1,1}(\Real)$ is the usual Sobolev space of integrable functions with integrable first derivative, and
%  $H^\ell(\Real)$ for $1\le \ell<\infty$ denotes the Sobolev space of functions with square integrable derivatives up to order $\ell$.
% \item
% $\L(\Real)$ denotes the function space
%   \be
%   \L(\Real)=\left\{f(x)\in C^1(\Real)| f'(x)\in\ L^1(\Real)\right\}.\label{e:spaceH}
%   \ee
%   {\color{red}Do we need continuity of $f,f'$? Would it be sufficient to require $d(k,z),p(k,z)\in W_k^{1,1}(\Real)$ for $z\ge 0$ and in addition $d'(k,z),p'(k,z)\in L^\infty(\Real\times \Real^+)$? See Prop. 4.6 below. If so, we might not need the space $\mathcal{L}$ and the following remark.}
\item
$L^{p,q}(\Real)$ denotes the weighted $L^p(\Real)$ space with norm
  \be
  \|f\|_{L^{p,q}}=\left(\int_{\Real}\langle x\rangle^{2q}|f(x)|^p \d x\right)^{1/p},
\ee
with $\langle x\rangle=(1+x^2)^{1/2}$.
\item
Similarly to \cite{DeiftZhou2003} and other works, our results will be formulated in the weighted Sobolev space
$H^{1,1}(\Real)$, where
\[
H^{1,1} (\Real) = \{f:\Real\to\Complex: f\in L^{2,1}(\Real) \cap H^1(\Real)\}\,.
\]
(To avoid confusion, however, we note that other works in the literature use $H^{1,1}(\Real)$ to denote the space of functions $f$ such that both $f$ and its derivative belong in $L^{2,1}(\Real)$.)
\item
Finally, for functions of several variables (e.g., $t$ and $k$, or $k$ and $z$, etc.), we will use a dot to specify the variable with respect to which the space is being considered.
For instance, for a function $f=f(t,z)$, we will use $f(\cdot,z)\in L^p(\Real)$ to signify that $f(t,z)\in L^p(\Real)$ as a function of $t$
for any fixed value of $z$ in a specified range.
\end{itemize}

%\begin{remark}
%If $f(k)\in\L(\Real)$,
%then $f'\in L^1(\Real)\cap L^\infty(\Real)$.
%%Moreover, the limits \,
%%$f(k)$ exists as $k\to \infty$. Moreover, if $f(k)\in C^1(\Real)$ and
%$\lim_{k\to\pm\infty} f(k)=f_{\pm\infty}$ exist and are finite,
%although they need not be equal.
%then $f\in\L(\Real)$.\marginpar{We know $f'\to 0$ as $k\to \infty$, but how do we know it decays faster than $1/k$?}
%\end{remark}

In \cite{Zhou1998}, Zhou established the $L^2$-Sobolev bijectivity result for the IST of the focusing NLS equation.
As mentioned above, the focusing NLS and the MBEs share the same scattering problem, i.e., Eq.~\eqref{e:ZS}.
Zhou used $q(x,t)$ to denote the potential, and its role in the MBEs is played by $q(t,z)$. Importantly, Zhou’s approach does not require avoiding spectral singularities or limiting the number of discrete eigenvalues of the scattering problem.
%Correspondingly, the ZS-AKNS system in \cite{Zhou1998} takes the form
%\be
%\partial_t \psi=ik\sigma_3\psi+Q(t)\psi,
%\ee
%which coincides with the first equation of \eqref{e:Laxpair}.
The key question is under which conditions the $L^2$-Sobolev bijectivity extends to the $z$-propagation in the MBEs,
and how this depends on the asymptotics of the density matrix as $t \to-\infty$, namely on the functions $P_-(z, k)$ and $D_-(z, k)$.

First, we express Zhou's bijectivity results insofar as they can be directly applied to the MBEs, with the above mentioned adaptation in the notations for the independent variables (i.e., replacing $x$ with $t$, and $t$ with $z$). We also mention that in Zhou's paper the spectral parameter is $k/2$, but this re-scaling bears no consequences on the extension of the results.

\begin{lemma}
 For a fixed $z\ge0$, if $q(\cdot,z)\in H^{1,1}(\Real)$, then the associated reflection coefficient defined in section~\ref{s:direct} $r(\cdot,z)\in H^{1,1}(\Real)$.
\end{lemma}

\begin{corollary}\label{c:LC1}
For a fixed $z\ge0$, let $q(\cdot,z)$, $\breve{q}(\cdot,z)\in H^{1,1}(\Real)$, such that $\|q(\cdot,z)\|_{H^{1,1} (\Real)}$, $\|\breve{q}(\cdot,z)\|_{H^{1,1} (\Real)}\leq U(z)$  with $U(z)>0$. Denote the corresponding reflection coefficients by $r(k,z)$ and $\breve{r}(k,z)$, respectively. Then, there is a positive $C(U,z)$ such that:
	\be
	\|r(\cdot,z)-\breve{r}(\cdot,z)\|_{H^{1,1} (\Real)}\leq C(U,z) \|q(\cdot,z)-\breve{q}(\cdot,z)\|_{ H^{1,1}(\Real)},
	\ee
	which means that for any fixed $z\ge 0$ the mapping:
	\be
	 q(\cdot,z)\in H^{1,1} (\Real) \to r(\cdot,z)\in H^{1,1} (\Real)
	\ee
	is Lipschitz continuous.
\end{corollary}

\begin{lemma}
For a fixed $z\ge 0$, let $r(\cdot,z)\in  H^{1,1} (\Real)$.  Then $q(\cdot,z)\in H^{1,1} (\Real)$ satisfies
	\be
	\|q(\cdot,z)\|_{H^{1,1} (\Real)}\leq C(z)\|r(\cdot,z)\|_{ H^{1,1} (\Real)},
	\ee
	where $q(t,z)$ is the optical pulse obtained from $r(k,z)$ via the reconstruction formula in section \ref{s:inverse}, and $C(z)>0$ depends on $\|r(\cdot,z)\|_{H^{1,1} (\Real)}$.
\end{lemma}

\begin{corollary}\label{c:LC2}
For fixed $z\ge0$, let $r(\cdot,z)$, $\breve{r}(\cdot,z)\in H^{1,1}(\Real)$ satisfy $\|r(\cdot,z)\|_{H^{1,1} (\Real)}$, $\|\breve{r}(\cdot,z)\|_{H^{1,1}(\Real)}\leq V(z)$ for some $V(z)>0$. Denote the corresponding potentials by $q(t,z)$ and $\breve{q}(t,z)$, respectively. Then, there is a $C(V,z)>0$ such that:
	\be
	\|q(\cdot,z)-\breve{q}(\cdot,z)\|_{H^{1,1}(\Real)}\leq C(V,z) \|r(\cdot,z)-\breve{r}(\cdot,z)\|_{H^{1,1} (\Real)},
	\ee
	which means that the mapping:
	\be
	r(\cdot,z)\in H^{1,1} (\Real) \to q(\cdot,z)\in H^{1,1}(\Real)
	\ee
	is Lipschitz continuous.
\end{corollary}

Next  we prove that, if the above results hold at $z=0$, they hold for any $z\in[0,Z]$, for a suitable $Z>0$ specified below.
\begin{lemma}
Let the boundary data $d(k,z)$ and $p(k,z)$ admit weak derivatives with respect to $k$, denoted respectively by $d'(k,z)$ and $p'(k,z)$,
and $d'(k,z),p'(k,z)\in L^\infty(\Real\times[0,Z])$ for some $Z>0$. If $q_o(t)=q(t,0)\in H^{1,1} (\Real)$, then $r(\cdot,z)\in  H^{1,1} (\Real)$ for all $z\in[0,Z]$.
\end{lemma}

\begin{proof}
Recall \eqref{e:rsolution}, namely
\be
r(k,z) = \e^{-w(k)\D(k,z)} \left( r_0(k) - \pi g(k) \int_{0}^z\e^{w(k)\D(k,\z)}P_-^*(k,\z)\,\d \z \right),
\ee
 with $r_0(k)=r(k,0)$,  from which it follows that
\begin{align}
	\label{e:rL21}
	\|r(\cdot,z)\|_{L^{2,1}(\Real)}^2 &=\int_{\mathbb{R}}\langle k\rangle^2|r(k,z)|^2\d k\nonumber
	\\
	&= \int_{\mathbb{R}}\langle k\rangle^2\left|\e^{-w(k)\D(k,z)}\right|^2\left| r_0(k) - \pi g(k) \int_{0}^z\e^{w(k)\D(k,\z)}P_-^*(k,\z)\,\d \z \right|^2\d k\nonumber
	\\
	&\leq \int_{\mathbb{R}}\langle k\rangle^2\e^{2|w(k)\D(k,z)|}\left| r_0(k) - \pi g(k) \int_{0}^z\e^{w(k)\D(k,\z)}P_-^*(k,\z)\,\d \z \right|^2\d k\nonumber
	\\
	&\leq \e^{2z/\epsilon} \int_{\mathbb{R}}\langle k\rangle^2\left|  r_0(k) - \pi g(k) \int_{0}^z\e^{w(k)\D(k,\z)}P_-^*(k,\z)\,\d \z  \right|^2\d k,
\end{align}
where we used the fact that $|\D(k,z)|\leq z$,
$|w(k)z|=z\big/{\sqrt{\epsilon^2+k^2}}\leq {z}/{\epsilon}$ for any $z\ge 0$,  and $|e^{y}| = e^{- \Re y} \leq e^{|\Re y|} \leq e^{|y|}$, for any $y\in \Complex$.
Moreover, we have
\be
\label{e:estimation1}
\left|\int_{0}^z\e^{w(k)\D(k,\z)}P_-^*(k,\z)\,\d \z\right|\leq \int_{0}^z \e^{|w(k)|\z}|P_-(k,\z)| \d \z \leq \frac{\e^{|w(k)|z}-1}{|w(k)|}\leq \epsilon(\e^{z/\epsilon}-1),
\ee
%since $w(k)\to 0$ as $k\to\infty$, using a Taylor expansion, we have for any fixed $z$:
%\be
% \frac{\e^{z\, w(k)\cos d(k)}-1}{w(k)\cos d(k)}=z+\frac{w(k)\cos d(k)z^2}{2}+O(1/k^2).
%\ee
%Hence,  for any fixed $z\geq0$, we have
where we have used that $|P_-(k,z)|\le 1$ $\forall k\in \Real$ and $z\in [0,Z]$. The last inequality in \eqref{e:estimation1} follows by noticing that $(e^{c x}-1)/x$ is an increasing function of $x\in \Real$, and in our case $0\le x=|w(k)|\le 1/\epsilon$.
%\marginpar{The other bound $|e^z-1|\le e^{|z|}-1$ follows from Taylor series.}

Using the bounds in \eqref{e:rL21}  and \eqref{e:estimation1},  as well as $r_0(k)=r(k,0)\in L^{2,1}(\Real)$ (which follows from Lemma~4.1 at $z=0$),  and $g(k)\in L^{2,1}(\Real)$ (cf. \eqref{e:g}),  we obtain
\be\label{e:rH111}
\|r(\cdot,z)\|_{L^{2,1}(\Real)} \leq \tilde{C}(z)\left(\| r_0 \|_{L^{2,1}(\Real)}+\|g\|_{L^{2,1}(\Real)}\right)<\infty,
\ee
 for some $\tilde{C}(z)>0$, proving that $r(\cdot,z)\in L^{2,1}(\Real)$  for any $z\in [0,Z]$.   Next, we are going to prove that $ r(\cdot,z)\in H^1(\Real)$ for any $z\in [0,Z]$.
From \eqref{e:rsolution}, we have
\be
r'(k,z)=\frac{\partial r(k,z)}{\partial k}:=r_{1}'(k,z)+r_{2}'(k,z),
\ee
where
\bse
\begin{gather}
	r_{1}'(k,z)=- \left( w' (k)\D(k,z)+w(k)\D'(k,z)\right) r(k,z),
\end{gather}
\vglue-2\bigskipamount
\begin{align}
    r_{2}'(k,z)=&\e^{-w(k)\D(k,z)}\left[r'_0(k)-\pi g'(k)\int_{0}^z\e^{w(k)\D(k,\z)}P_-^*(k,\z)\,\d \z \right.\nonumber
\\
	&- \pi g(k)\int_{0}^z\e^{w(k)\D(k,\z)}\left( w' (k)\D(k,\z)+w(k)\D'(k,\z)\right) P_-^*(k,\z)\,\d \z\nonumber
\\
    &\left.+\pi g(k)\int_{0}^z\e^{w(k)\D(k,\z)}(P_-'(k,\z))^{*}\,\d \z\right],
\end{align}
\ese
and
\bse
\begin{gather}
w'(k)=\frac{\i}{\epsilon^2+k^2}-\frac{2k(\i k-\epsilon)}{(\epsilon^2+k^2)^2}, \qquad
g'(k)=-\frac{2\epsilon k}{\pi (\epsilon^2+k^2)^2}\,,
\\
\D'(k,z)=-\int_0^z d'(k,\z)\sin d(k,\z) \d \z,\\
(P_-'(k,z))^*=d'(k,z)\cos d(k,z)\e^{-\i p(k,z)}-\i p'(k,z)\sin d(k)\e^{-\i p(k,z)}.
\end{gather}
\ese
[Recall that prime denotes derivative with respect to the spectral parameter~$k$ throughout.]
From \eqref{e:w}, it is obvious that $w(k)\in H^2(\Real)$ and that $w(k)$, $w'(k)\in L^\infty(\Real)$. Also, since  $d'(k,z)\in L^\infty(\Real\times [0,Z])$,
%and is bounded with respect to $z$,
we have
\bse
\begin{gather}
|\D'(k,z)|\leq \int_0^z |\sin d(k,\z)| |d'(k,\z)|\d \z \leq \int_0^z |d'(k,\z)|\d \z \leq c(Z) z \\
c(Z)=\sup_{k\in \Real,z\in [0,Z]}|d'(k,z)|\,.
\end{gather}
\ese
Thus,  $w' (k)\D(\cdot,z)+w(k)\D'(\cdot,z) \in L^\infty(\Real)$, so that for any $z\in [0,Z]$:
\be\label{e:rH112}
\|r_{1}'(\cdot,z)\|_{L^2(\Real)} \leq z\left(\|w'\|_{\infty}+c(Z)\|w\|_{\infty}\right)\|r(\cdot,z)\|_{L^2(\Real)}\leq \frac{z}{\epsilon}\left(1/\epsilon+c(Z)\right)\|r(\cdot,z)\|_{L^2(\Real)},
\ee
where we used the facts that $\|w\|_{L^{\infty}}\leq 1/\epsilon$ and $\|w'\|_{L^{\infty}}\leq 1/\epsilon^2$.
For $r_{2}'(k,z)$, we have  $\left|\e^{w(k)\D(k,z)}\right|\leq \e^{z/\epsilon}$, $r_0(k)\in H^1(\Real)$,  $|P'(k,z)|\leq \|p'\|_{L^\infty(\Real\times [0,Z])}+\|d'\|_{L^\infty(\Real\times [0,Z])}$, and \eqref{e:estimation1}.
We also have
\begin{align}\label{e:rH113}
&\left|\int_{0}^z\e^{w(k)\D(k,\z)}\left( w' (k)\D(k,z’)+w(k)\D'(k,z‘)\right) P_-^*(k,\z)\,\d \z\right|\nonumber
\\
&\leq  \int_{0}^z\left|\e^{w(k)\D(k,\z)}\right|\left| w' (k)\D(k,z’)+w(k)\D'(k,\z)\right| \d \z\nonumber
\\
&\leq \int_0^z \e^{\z/\epsilon} \|w'\|_{\infty} \z\d \z +\int_0^z c(Z)\|w\|_{\infty} \e^{\z/\epsilon} \z\d \z\nonumber
\\
&=\left(\|w'\|_\infty+c(Z)\|w\|_\infty\right)\left(\epsilon z\e^{z/\epsilon}-\epsilon^2\e^{z/\epsilon}+\epsilon^2\right)\nonumber
\\
&\le  \left( 1+\epsilon c(Z)\right)\left(\frac{Z}{\epsilon}\e^{Z/\epsilon}+\e^{Z/\epsilon}+1\right).
\end{align}
Since $g(k)$, $g'(k)\in L^2(\Real)$, using \eqref{e:estimation1} it follows that $r'_{2}(\cdot,z)\in L^2(\Real)$ for all $z\in [0,Z]$, which therefore implies that $r(\cdot,z) \in H^1(\mathbb R)$. Consequently, $r(\cdot,z)\in H^{1,1}(\Real)$ for every $z\in [0,Z]$.
\end{proof}

\begin{theorem}
\label{t:localwellposedness}
(Local well-posedness)~
Let the initial datum $q(t,0)=q_0(t)\in H^{1,1} (\Real)$.
If the boundary data $d(k,z),p(k,z)$ admit weak derivatives  $d'(k,z),p'(k,z)$  with respect to $k$,
and $d'(k,z),p'(k,z)\in L^\infty(\Real\times [0,Z])$ for some $Z>0$,
then for each $z\in [0, Z]$ there exists a unique local solution $q(\cdot,z)\in H^{1,1}(\Real)$ to the Cauchy problem~\eqref{e:MBE}.
Moreover, the map
\be
    q(\cdot,0)\in H^{1,1} (\Real) \mapsto q(\cdot,z)\in H^{1,1}(\Real)
\label{e:istmap}
\ee
is Lipschitz continuous.
\end{theorem}

\begin{proof}
Recall that the potential $q(t,z)$ is recovered from the scattering data $r(k,z)$ with the inverse scattering transform as in \cite{Zhou1998}.
Thus, $q(\cdot,z)$ is defined in $H^{1,1} (\Real)$ for every $z\in [0,Z]$ and is a Lipschitz continuous function of $r(k,z)$.  Let $r_0(k)=r(k,0)$, and let the positive  quantities  $c_1$, $c_2$, $c_3$ and $c_4$ depend, respectively, on
$\|r(\cdot,z)\|_{H^{1,1}(\Real)}$,
$(Z, \|r_0\|_{H^{1,1}(\Real)})$,
$(Z, \|g\|_{L^{2,1}(\Real)},\|d'\|_{L^\infty(\Real\times [0,Z])},\|p'\|_{L^\infty(\Real\times [0,Z])})$ and
$(Z, \|q_0\|_{H^{1,1} (\Real)})$.  For all $z\in[0,Z]$ we have
\be\label{qtz}
	\|q(\cdot ,z)\|_{H^{1,1}(\Real)}\leq c_1\|r(\cdot,z)\|_{H^{1,1} (\Real)}\leq c_2\|r_0\|_{H^{1,1}(\Real)}+c_3\leq c_4\|q_0\|_{H^{1,1} (\Real)}+c_3\,,
\ee
 where the second inequality follows from  \eqref{e:rH111}, \eqref{e:rH112} and  \eqref{e:rH113},  which exclude blow-up in a finite time and enable application of Zhou's result on the bijectivity between the solutions to the Maxwell-Bloch system and the reflection coefficients in IST.
	
Finally, the Lipschitz continuity of the map \eqref{e:istmap} for any $z\in[0,Z]$ follows from Corollary \ref{c:LC2}.
\end{proof}

The following theorem shows that there exists a global solution in $H^{1,1}(\Real)$.

\begin{theorem}
\label{t:globalwellposedness}
(Global well-posedness)~
Let the initial datum $q(t,0)=q_0(t)\in H^{1,1}(\Real)$.
If the boundary data $d(k,z),p(t,z)$ admit weak derivatives $d'(k,z),p'(k,z)$ with respect to $k$,
and
$d'(k,z),p'(k,z)\in L^\infty(\Real\times [0,\infty))$,
then there
	exists a unique global  in $z$ solution  $q(\cdot,z)\in H^{1,1}(\Real)$ to the Cauchy problem for the MBEs
	\eqref{e:MBE} with initial-boundary conditions~\eqref{e:ICBC}.
	Moreover, the map
\be
	 q(\cdot,0)\in H^{1,1}(\Real) \mapsto q(\cdot,z)\in C\big(H^{1,1}(\Real)\times[0,\infty)\big)
\ee
is Lipschitz continuous.
\end{theorem}

\begin{proof}
Suppose there is a maximal value $Z_{\rm max}>0$ for which the local solution exists.
If $Z_{\rm max}=\infty$, then the local solution is also the global solution and the result holds.
If $Z_{\rm max}<\infty$ and the local solution exists in the closed interval $[0,Z_{\rm max}]$, we can use $q(\cdot,Z_{\rm max})\in H^{1,1} (\Real)$ as the new initial data. Following the same inverse scattering transform as in \cite{Zhou1998}, there exists a positive constant $Z_1$, such that
$q(t,z)\in C\big(H^{1,1} (\Real)\times [Z_{\rm max}, Z_{\rm max}+Z_1]\big)$.
This implies a contradiction with the maximal value assumption.
	
Finally, if the local solution exists in the half open interval $[0,Z_{\rm max})$, from \eqref{qtz} we have
\be
\label{qtzZ}
	 \|q(\cdot,z)\|_{H^{1,1}(\Real)}\leq \tilde{c}_4(Z_{\rm max})\|q_0\|_{H^{1,1} (\Real)}+\tilde{c}_3(Z_{\rm max}), \quad z\in[0,Z_{\rm max}),
\ee
where $\tilde{c}_4(z)$ and $\tilde{c}_3(z)$ may grow at most polynomially in $z$  but they remain finite for every $z>0$.
Due to the continuity of $q(t,z)$ with respect to $z$, the limit of $q(t,z)$ as $z\to Z_{\rm max}$ exists. Here we denote the limit by $q_{\rm max}(t)$. Taking the limit of \eqref{qtzZ} as $z\to Z_{\rm max}$, we have
\be
	 \|q_{\rm max}\|_{H^{1,1}(\Real)}\leq \tilde{c}_4(Z_{\rm max})\|q_0\|_{ H^{1,1}(\Real)}+\tilde{c}_3(Z_{\rm max}),
\ee
which implies that the local solution
$q(t,z)\in C( H^{1,1} (\Real)\times [0,Z_{\rm max}))$
can be extended to
$q(t,z)\in$\break$C( H^{1,1} (\Real)\times [0,Z_{\rm max}])$,
which contradicts the assumption that $[0,Z_{\rm max})$ is the maximal open interval of existence.
This completes the proof that the local solution can be extended to a global one.
\end{proof}

%%%%%%%%%%%%%%%%%%%%%%%%%%%%%%%%%%%%%%%%%%%%%%%%%%%%%%%%%%%%%%%%%%%%%%%%%%%%%%%%%%%%%%%%%%%%%%%%%%%%%%
\section{Asymptotic states of propagation}
\label{s:asymp}

In this section we discuss the long-distance asymptotics of the solutions of the MBEs~\eqref{e:MBE}.
It should already be clear from the previous section that the behavior will be heavily dependent on the value of $D_-$.
Therefore one must study several cases separately.
%For simplicity, we limit this discussion to the case in which $D_-$ and $P_-$ are constant, independent of $z$ and~$k$,
%but similar considerations apply in more general cases.

%%%%%%%%%%%%%%%%%%%%%%%%%%%%%%%%%%%%%%%%%%%%%%%%%%%%%%%%%%%%%%%%%%%%%%%%%%%%%%%%%%%%%%%%%%%%%%%%%%%%%%
\paragraph{Asymptotic value of the scattering coefficients.}

We begin by looking at the asymptotic value of the reflection coefficient for large $z$.
Recall that the evolution (i.e., propagation inside the medium) of the reflection coefficient $r(k,z)$ as a function of $z$ is given by \eqref{e:rsolution}.
Therefore, its behavior as a function of $z$ is determined by the sign of the real part of $w(k,z)$, which is given by \eqref{e:w}.
Since $D_-(k,z)$ and $g(k)$ are real-valued, and $g(k)$ is positive, the growth or decay of $r(k,z)$ is completely determined by the sign of $D_-(k,z)$.

Consider first the case in which $P_-(k,z)$ is identically zero (i.e., $D_-(k,z)= \pm1$ $\forall k\in\Real$ and $\forall z\ge0$).
Inspection of \eqref{e:rkz} shows that if the medium is initially in the stable pure state (i.e., it is prepared so that $P_-=0$ and $D_-=-1$), $r(k,z)$ is exponentially decaying as $z\to+\infty$ for all $k\in\mathbb{R}$.
Conversely, if the medium is initially in the unstable pure state (i.e., it is prepared so that $P_-=0$ and $D_-=1$), then $r(k,z)$ is exponentially growing as $z\to+\infty$ for all $k\in\mathbb{R}$.
Similar considerations can be made when $P_-(k,z)$ is not identically zero,
but in this case the analysis identifies several different cases.
Let us consider the scenario in which $D_-$ and $P_-$ are independent of~$z$ for simplicity, in which case
$r(k,z)$ is given by~\eqref{e:rkz}.
Summarizing, inspection of~\eqref{e:rkz} shows the following:
\begin{proposition}
Assume that $D_-,P_-$ are independent of $z$ and let $r(k,z)$ be the reflection coefficient as given by~\eqref{e:rkz}.
	\vspace*{-1ex}
	\begin{itemize}
		\item[(i)]
		If $D_-(k)=-1$, $r(k,z) = \e^{w(k)z}r(k,0)$, with $w(k)$ given by~\eqref{e:w}.
[Recall that according to \eqref{e:w}, $w(k)z\equiv \frac{ik-\epsilon}{k^2+\epsilon^2}z$]		
\item[(ii)]
		If $D_-(k)=1$, $r(k,z) = \e^{-w(k)z}r(k,0)$, with $w(k)$ given by~\eqref{e:w}.
		\item[(iii)]
		If $D_-(k) = 0$, $r(k,z)$ grows linearly in~$z$.  Explicitly, $r(k,z)= r(k,0) - \pi g(k) P_-^* z$.
		\item[(iv)]
		If\  $-1<D_-(k)<0$, $\lim_{z\to+\infty} r(k,z) = \epsilon P_-^*(k)/(\i k-\epsilon)$.
		\item[(v)]
		If\  $0<D_-(k)<1$, $r(k,z)$ exhibits the same kind of exponential growth as when $D_-=1$.  Namely,
		\be
		\lim_{z\to +\infty} \e^{w(k)z}r(k,z) = r(k,0) \,.
		\ee
	\end{itemize}
\end{proposition}
It should be clear that similar considerations can apply in more general cases when $D_-$ and/or $P_-$ depend on $z$.

%%%%%%%%%%%%%%%%%%%%%%%%%%%%%%%%%%%%%%%%%%%%%%%%%%%%%%%%%%%%%%%%%%%%%%%%%%%%%%%%%%%%%%%%%%%%%%%%%%%%%%
\paragraph{Asymptotic state of the medium.}

Next, we look at the asymptotic state of the medium as $t\to+\infty$, as given by $D_+$ and $P_+$, which are determined by the reflection coefficient $r(k,z)$ in \eqref{e:D+P+}.
Consider first the case in which the medium is initially in the stable pure state (i.e., $P_-=0$ and $D_-=-1$).
In this case, since $r(k,z)$ decays exponentially as $z\to+\infty$ for all $k\in\Real$, \eqref{e:D+P+} implies that $D_+\to -1$ and $P_+\to 0$ for large $z$. Therefore, the medium returns to the stable state for sufficiently large propagation distances, justifying the use of the term ``stable state''.
Conversely, if the medium is initially prepared in the unstable pure state (i.e., $P_-=0$ and $D_-=1$), $r(k,z)$ is exponentially growing as $z\to+\infty$ for all $k\in\Real$, and \eqref{e:D+P+} still gives $D_+\to -1$ and $P_+\to 0$ for large~$z$. Therefore, the medium reverts to the stable state for sufficiently large propagation distances.
Similar considerations can be made when $P_-$ is not identically zero.
Summarizing, inspection of~\eqref{e:D+P+} shows the following:
\begin{proposition}
	Let $D_+(k,z)$ and $P_+(k,z)$ be given by~\eqref{e:D+P+}.
	\vspace*{-1ex}
	\begin{itemize}
		\item[(i)]
		If $D_-(k) = -1$, $D_+(k,z)\to -1$ and $P_+(k,z)\to 0$ as $z\to +\infty$.
		\item[(ii)]
		If $D_-(k) = 1$, $D_+(k,z)\to -1$ and $P_+(k,z)\to 0$ as $z\to+\infty$.
		\item[(iii)]
		If $D_-(k) = 0$, $D_+(k,z)\to 0$ and $|P_+(k,z)|\to |P_-(k,z)|$ as $z\to+\infty$.
		\item[(iv)]
		If \ $-1<D_-(k)<0$,
		\be
		\lim_{z\to+\infty} D_+(k,z) =  \frac{\epsilon^2(D_-(k))^3+k^2D_-(k)-2\epsilon^2+2(D_-(k))^2\epsilon^2}{\epsilon^2(2-(D_-(k))^2)+k^2}.
		\ee
		\item[(v)]
		If $0<D_-(k)<1$, $D_+(k)\to -D_-(k)$ and $|P_-(k)|\to|P_+(k)|$ as $z\to+\infty$.
	\end{itemize}
\end{proposition}
%In the next section we will use the above considerations to prove the well-posedness of the MBEs~\eqref{e:MBE}.

%%%%%%%%%%%%%%%%%%%%%%%%%%%%%%%%%%%%%%%%%%%%%%%%%%%%%%%%%%%%%%%%%%%%%%%%%%%%%%%%%%%%%%%%%%%%%%%%%%%%%%
\section{Concluding remarks}
\label{s:concl}

In summary, taking advantage of the $L^2$-Sobolev bijectivity of the IST for the focusing NLS equation proved by Zhou,
we have established the local and global well-posedness of the MBEs~\eqref{e:MBE} with inhomogenous broadening.
Importantly, the bounds in Lemma~4.5 become singular as $\epsilon\to 0^+$, and therefore the results cannot be extended to the sharp-line limit,
consistently with the findings of \cite{LiMiller}.
The results of this work clearly indicate that, when the physical effect of inhomogenous broadening is taken into account,
the corresponding MBEs are more well-behaved than in the singular case of the sharp-line limit.

The results of this work also fit in the context of a long list of studies of well-posedness for nonlinear wave equations.
For example, sharp well-posedness of the NLS equation on the line with initial data in Sobolev spaces $H^s$ for any $s\geqslant 0$ was proved by  Bourgain \cite{Bourgain1993} (see also \cite{Bourgain1999}).
Well-posedness of the NLS equation on the half-line with data in Sobolev spaces was established by Holmer \cite{h2005},
Bona, Sun and Zhang \cite{bsz2015} and
Fokas, Himonas and Mantzavinos \cite{fhm2015}.
Further functional-analytic results for the NLS equation were obtained by
Craig, Kappeler and Strauss \cite{cks1995},
Cazenave \cite{Cazenave2003},
Cazenave and Weissler \cite{cw2003},
Ghidaglia and Saut \cite{gs1993},
Ginibre and Velo \cite{gv1979},
%Carroll and Bu \cite{cb1991},
Kenig, Ponce and Vega \cite{kpv1991},
Kato \cite{k1995},
Linares and Ponce \cite{lp2009},
Tsutsumi \cite{t1987}.
(For the Korteweg-deVries and modified Korteweg-deVries equations, see \cite{CKSTT2003} and references therein.)

For the derivative NLS equation,
global well-posedness results were obtained by
Colliander, Keel, Staffiilani, Takaoka and Tao in \cite{CKSTT2002}, Wu \cite{Wu2013,Wu2015} and Guo and Wu \cite{GuoWu2017}.
More recently, using the IST without discrete eigenvalues or resonances, in~\cite{PelinovskyDNLS1} Pelinovsky \textit{et ~alii} established
the existence of global solutions to the derivative NLS equation without any small norm assumption.
One of the key points in~\cite{PelinovskyDNLS1} is the introduction of a transformation of the scattering problem to
a spectral problem of ZS-type.
Using an invertible Backl\"und transformation, the authors then studied the global well-posedness of the derivative NLS equation in the case when the initial data includes a finite number of  solitons.
In the context of the present work, since the scattering problem of the MBEs is already of ZS-type,
the methods utilized in \cite{PelinovskyDNLS1} would not bring any improvement compared to Zhou's results.

In~\cite{Bahouri},
Bahouri and Perelman showed that the IVP for the derivative NLS equation is globally well-posed for general Cauchy data in $H^{1/2}(\Real)$ and that, furthermore, the $H^{1/2}$ norm of the solutions remains globally bounded in time.
This result closes the discussion in the setting of the Sobolev spaces $H^s$.
Most recently, Harrop-Griffiths, Killip, Ntekoume and Visan proved that the derivative NLS equation in one space dimension is
globally well-posed on the line in $L^2(\Real)$, which is the scaling-critical space for this equation~\cite{L2wellposednessDNLS}.
The results of the present work raise the natural question of whether the well-posedness of the MBEs~\eqref{e:MBE}
can also be established in more general spaces.

Finally, the rigorous calculation of the long-distance asymptotic behavior of the optical pulse with various choices of
medium preparation is also a very interesting, physically relevant open problem, which is left for future study.

%%%%%%%%%%%%%%%%%%%%%%%%%%%%%%%%%%%%%%%%%%%%%%%%%%%%%%%%%%%%%%%%%%%%%%%%%%%%%%%%%%%%%%%%%%%%%%%%%%%%%%
\bigskip\noindent\textbf{\sffamily Acknowledgments}
\par\medskip\noindent
The authors would like to sincerely acknowledge D. Mantzavinos and D. Pelinovsky %and G. Staffilani
for their valuable feedback on the manuscript.
Partial support for this work from the U.S.\ National Science
Foundation, under grants DMS-2009487 (GB), DMS-2106488 and DMS-2406626 (BP), is gratefully acknowledged.

\section*{Appendix}
\setcounter{section}1
\setcounter{subsection}0
\setcounter{equation}0
\def\thesection{\Alph{section}}
\def\theequation{\Alph{section}.\arabic{equation}}
\def\thetheorem{\Alph{section}.\arabic{theorem}}
\def\thefigure{\Alph{section}.\arabic{figure}}
\addcontentsline{toc}{section}{Appendix}

In this appendix, we consider the Cauchy problem for the NLS equation:
\be
\i q_z+q_{tt}+2|q|^2q=0,\quad q(t,0)=q_0(t).
\label{e:NLS}
\ee
and we show that the $L^2$-bijectivity of the IST, which was proved in Zhou's paper \cite{Zhou1998},
remains preserved by the time evolution, which in turn implies local and global well-posedness,
a detail not explicitly addressed in Zhou's work.
\begin{proposition}
Let the initial datum $q_0(t)\in H^{1,1}(\Real)$.
Then there exists a unique global  in $z$ solution $q(\cdot,z)\in H^{1,1}(\Real)$ to the Cauchy problem for the NLS equation
with initial-boundary conditions~\eqref{e:NLS}.
Moreover, the map
\be
q(t,0)\in H^{1,1}(\Real) \mapsto q(\cdot,z)\in C\big(H^{1,1}(\Real)\times[0,\infty)\big)
\ee
is Lipschitz continuous.	
\end{proposition}
\begin{proof}
	Here we just need to prove that the time-dependent reflection coefficient $r(k,z)=\e^{2\i k^2 z} r(k,0)\in H^{1,1}(\Real)$. According to \cite{Zhou1998}, it is established that $r(k,0)\in H^{1,1}(\Real)$. Subsequently, $r(\cdot,z)\in L^{2,1}(\Real)$. To prove $r'(\cdot,z)=4\i k\e^{2\i k^2t}r(k,0)+\e^{2\i k^2t}r'(k,z)\in L^2(\Real)$, we notice that $r(k,0)\in L^{2,1}(\Real)$ and $r'(k,0)\in L^2(\Real)$. Therefore, $r'(\cdot,z)\in L^2(\Real)$, which completes the proof.
\end{proof}

\noindent
The same methods as for Theorems \ref{t:localwellposedness} and \ref{t:globalwellposedness} then yield:

\begin{theorem}
\label{t:NLSlocalwellposedness}
(Local well-posedness)~
Let the initial datum $q(t,0)=q_0(t)\in H^{1,1} (\Real)$.
For each $z\in [0, Z]$ there exists a unique local solution $q(\cdot,t)\in H^{1,1}(\Real)$ to the Cauchy problem~\eqref{e:NLS}.
Moreover, the map
\be
    q(\cdot,0)\in H^{1,1} (\Real) \mapsto q(\cdot,t)\in H^{1,1}(\Real)
\label{e:istmap}
\ee
is Lipschitz continuous.
\end{theorem}

\begin{theorem}
\label{t:NLSglobalwellposedness}
(Global well-posedness)~
Let the initial datum $q(t,0)=q_0(t)\in H^{1,1}(\Real)$.
There
exists a unique global in time solution  $q(\cdot,t)\in H^{1,1}(\Real)$ to the Cauchy problem~\eqref{e:NLS} for the NLS equation
Moreover, the map
\be
	 q(\cdot,0)\in H^{1,1}(\Real) \mapsto q(\cdot,t)\in C\big(H^{1,1}(\Real)\times[0,\infty)\big)
\ee
is Lipschitz continuous.
\end{theorem}

Of course $H^{1,1}(\Real)$ is not the optimal space for the NLS equation.
Moreover, the well-posedness of the Cauchy problem for the NLS equation in $H^1(\Real)$ had been proven with PDE methods
prior to Zhou's work,
see for example~\cite{lp2009} and references therein.
Still, we hope the above discussion serves to clarify how Zhou's bijectivity results for the IST also lead directly to well-posedness for the NLS equation.

%%%%%%%%%%%%%%%%%%%%%%%%%%%%%%%%%%%%%%%%%%%%%%%%%%%%%%%%%%%%%%%%%%%%%%%%%%%%%%%%%%%%%%%%%%%%%%%%%%%%%%
\medskip
\let\em=\it

\makeatletter
\def\@biblabel#1{#1.}
\small
\def\journal#1#2{\textit{#1}\unskip~\textbf{\ignorespaces #2}}
\def\v#1{\textbf{#1}}

\def\reftitle#1{``#1''}
\let\title=\reftitle

\end{document}